\documentclass[aps,pra,preprint,nofootinbib]{revtex4-1}
\usepackage{amsmath}
\usepackage{amssymb}
\usepackage{amsthm}
\usepackage{mathrsfs}
\usepackage{mathtools}

\newtheorem{definition}{Definition}
\newtheorem{theorem}{Theorem}
\newtheorem{corollary}{Corollary}[theorem]
\newtheorem{lemma}[theorem]{Lemma}

\def\dub#1{\underline{\underline{#1}}}

\begin{document}

\title{Coupled Kohn-Sham equations for electrons and phonons}

\author{Chung-Yu Wang}
\author{T. M\"{u}ller}
\author{S. Sharma}
\author{E. K. U. Gross}
\author{J. K. Dewhurst}
\affiliation{
Max-Planck-Institut f\"{u}r Mikrostrukturphysik, Weinberg 2, D-06120 Halle, Germany
}
\date{\today}

\begin{abstract}
This work establishes the algebraic structure of the Kohn-Sham equations
to be solved in a density formulation
of electron and phonon dynamics, including the superconducting
order parameter. A Bogoliubov transform is required
to diagonalize both the fermionic and bosonic
Kohn-Sham Hamiltonians since they both represent a non-interacting
quantum field theory.
The Bogoliubov transform for phonons is non-Hermitian in the general case,
and the corresponding time-evolution is non-unitary. Several sufficient
conditions for ensuring that the bosonic eigenvalues are real are provided and
a practical method for solving the system is described.
Finally, we produce a set of approximate mean-field potentials which are
functionals of the electronic and phononic density matrices and depend on the
electron-phonon vertex.
\end{abstract}

\maketitle

In this work we determine time-dependent Kohn-Sham matrix equations used
for combined systems of electron and phonons. Ultimately, the
potentials which enter the equations are considered to be functionals
of the density matrices produced from the time-evolving Kohn-Sham state.
One particular aim of this work is to include lattice degrees of freedom in simulations
of intense laser pulses acting on solids. This is necessary for the recovery
of the magnetic moment or the superconducting order parameter which are
typically destroyed by the laser pulse.

\section{Densities of the electron-nuclear system}
Consider the electron-nuclear Schr\"{o}dinger equation in atomic units:
\begin{align}\label{en_se}
 \hat{H}=\frac{1}{2}\sum_{i}\nabla_i^2
 +\sum_{I=1}\frac{1}{2M_I}\nabla_I^2
 +\sum_{i>j}\frac{1}{|{\bf r}_i-{\bf r}_j|}
 +\sum_{i,I}\frac{Z_I}{|{\bf r}_i-{\bf R}_I|}
 +\sum_{I>J}
 \frac{Z_I Z_J}{|{\bf R}_I-{\bf R}_J|}
\end{align}
for $i,j=1\ldots N_{\rm e}$ electrons
and $I,J=1\ldots N_{\rm n}$ nuclei, where $M_I$ is the nuclear
mass and $Z_I$ is the nuclear charge, assumed negative.
The wave function $\Psi(\dub{\bf r},\dub{s},\dub{\bf R},\dub{S},t)$,
where $\dub{s}$ and
$\dub{S}$ are electron and nuclear spin coordinates,
is determined in a finite (but large) box with periodic boundary conditions.

Conventional densities obtained from this wave function are spatially
constant and therefore not useful as variational quantities and a
different approach to density functional theory (DFT) is required.
The electron-nuclear wave function can be factored exactly\cite{Abedi2010} as:
\begin{align}\label{wf_fact}
 \Psi(\dub{\bf r},\dub{s},\dub{\bf R},\dub{S},t)=
 \Phi_{\dub{\bf R},\dub{S}}(\dub{\bf r},\dub{s},t)
 \chi(\dub{\bf R},\dub{S},t),
\end{align}
where
$\sum_{\dub{s}}\int d\dub{\bf r}\,|\Phi_{\dub{\bf R},\dub{S}}(\dub{\bf r},\dub{s},t)|^2=1$
for all $\dub{\bf R}$, $\dub{S}$ and $t$.

Let $V_{\rm BO}(\dub{\bf R})$ be the Born-Oppenheimer (BO) potential energy
surface (PES)\footnote{The BO PES is defined to be the ground state electronic
eigenvalue obtained from (\ref{en_se}) where the nuclear kinetic operator
is removed and the dependence on $\dub{\bf R}$ is parametric.}
and suppose this has a unique minimum at $\dub{\bf R}^0$.

\subsection{Electronic densities}
A purely electronic wave function is obtained by evaluating
$\Phi_{\dub{\bf R}^0\dub{S}}(\dub{\bf r},\dub{s},t)$. From this, a variety of
familiar electronic densities may be obtained, for example
\begin{align}\label{rho_0}
 \rho_{\dub{\bf R}^0}({\bf r},t)\equiv
 \sum_{\dub{S}}\int d^3r_2\ldots d^3r_{N_{\rm e}}
 \left|\Phi_{\dub{\bf R}^0,\dub{S}}(\dub{\bf r},\dub{s},t)\right|^2,
\end{align}
with similar definitions for the magnetization ${\bf m}({\bf r})$,
current density ${\bf j}({\bf r})$,
superconducting order parameter, $\chi({\bf r},{\bf r}')$ and so on.
Such a density is plotted in
Fig. \ref{hydrogen} for the hydrogen atom using various masses. Note that
this density is not a constant and also varies with the nuclear mass. The
densities for $M=\infty$ and the physical mass of a proton, $M\simeq 1836$, are
indistinguishable. However, the density is considerably different when the
nuclear and electronic masses are the same, $M=1$.
In the same figure is a plot of the density evaluated at a particular point
against $1/M$. The density decreases monotonically with
reciprocal mass and has a non-zero derivative at $1/M=0$.

A Kohn-Sham Hamiltonian defined to reproduce the density in
(\ref{rho_0}) as its ground state can be written as
\begin{align}\label{KS_fm_r}
 \hat{H}_{\rm KS}=-\frac{1}{2}\nabla^2+V_{\dub{\bf R}^0}({\bf r})
 +V_{\rm H}({\bf r})+V_{\rm xc}({\bf r})
 +V_{\rm fmc}({\bf r},t),
\end{align}
where $V_{\dub{\bf R}^0}({\bf r})$ is the external potential
determined from the nuclei fixed at $\dub{\bf R}^0$;
$V_{\rm H}$ and $V_{\rm xc}$ are the usual Hartree and
exchange-correlation potential; and $V_{\rm fmc}$ is a correction term
to account for the finite mass of the nuclei. Note that this potential
vanishes in the infinite mass limit, i.e.
$\lim_{M\rightarrow\infty}V_{\rm fmc}({\bf r},t)=0$, and the regular
Kohn-Sham equations for a fixed external potential are recovered.
The finite mass correction potential is plotted in Fig. \ref{hydrogen} for
hydrogen with an artificially light $M=2$. Not surprisingly, the potential is
mainly repulsive.
Mass correction potentials corresponding to other densities can also be defined
such as a magnetic field ${\bf B}_{\rm fmc}({\bf r},t)$
or a pairing potential $\Delta_{\rm fmc}({\bf r},{\bf r}',t)$. In the latter
case, the finite mass correction constitutes the entire potential for
phonon-coupled superconductors.

\begin{figure}[ht]
\centerline{\includegraphics[width=0.95\textwidth]{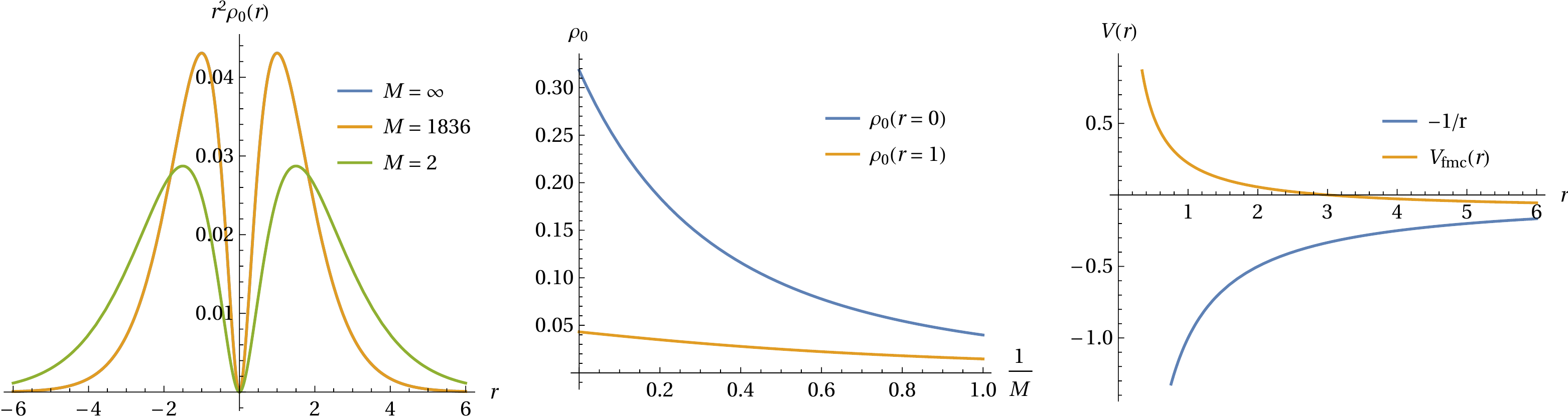}}
\caption{On the left is a plot of the electronic charge density times $r^2$,
as defined in (\ref{rho_0}), versus $r$ for various nuclear masses.
In the middle is the charge density evaluated at $r=0$ and $r=1$
plotted as a function of $1/M$.
On the right is a plot of the finite mass correction potential,
evaluated for $M=2$, plotted
alongside the nuclear potential $-1/r$.}\label{hydrogen}
\end{figure}

\subsection{Phonon densities}
We now consider the expansion of the BO PES around $\dub{\bf R}^0$ and
assume that the leading order, apart from a constant, is quadratic:
\begin{align}
 V_{\rm BO}(\dub{\bf R})=V_{\rm BO}(\dub{\bf R}^0)
 +\frac{1}{2}\sum_{I\alpha,J\beta}u_{I\alpha}
 K_{I\alpha,J\beta}u_{J\beta}+\cdots
\end{align}
where
$K_{I\alpha,J\beta}\equiv
 \left.\partial^2V_{\rm BO}/\partial R_{I\alpha}
 \partial R_{J\beta}\right|_{\dub{\bf R}^0}$,
$\dub{\bf u}\equiv \dub{\bf R}-\dub{\bf R}^0$
and $\alpha$, $\beta$ represent Cartesian directions.
The associated classical modes, called phonons, are determined by solving
the eigenvalue equation
\begin{align}\label{evphn}
 K{\bf e}_n=\nu_n^2 M{\bf e}_n
\end{align}
for $\nu_n$ and ${\bf e}_n$,
where $M_{I\alpha,J\beta}\equiv M_I\delta_{IJ}\delta_{\alpha\beta}$
is the diagonal matrix of nuclear masses.
Let $\hat{p}_{I\alpha}\equiv-i\partial_{I\alpha}$ be the momentum operator
which acts on a particular nuclear coordinate, then
$[\hat{u}_{I\alpha},\hat{p}_{J\beta}]=i\delta_{IJ}\delta_{\alpha\beta}$.
We can also define
\begin{align}
 \hat{\mathcal{U}}\equiv \mathcal{S}\hat{\bf u} \qquad
 \hat{\mathcal{P}}\equiv \mathcal{T}\hat{\bf p},
\end{align}
where $\mathcal{S}=2^{-\frac{1}{2}}\nu^{\frac{1}{2}}{\bf e}^t$,
$\mathcal{T}=2^{-\frac{1}{2}}\nu^{-\frac{1}{2}}{\bf e}^t M^{-1}$
and $\nu$ is the diagonal matrix of eigenvalues,
then $[\hat{\mathcal{U}},\hat{\mathcal{P}}]=\frac{i}{2}I$ and
$\hat{H}^{\rm b}=\hat{\mathcal{P}}^t\nu\hat{\mathcal{P}}
 +\hat{\mathcal{U}}^t\nu\hat{\mathcal{U}}$.
Writing
\begin{align}
 \hat{d}=\hat{\mathcal{U}}+i\hat{\mathcal{P}} \qquad
 \hat{d}^{\dag}=\hat{\mathcal{U}}^t-i\hat{\mathcal{P}}^t,
\end{align}
the Hamiltonian is cast in diagonal form
\begin{align}
 \hat{H}^{\rm b}=\sum_i\nu_i\left(\hat{d}_i^{\dag}\hat{d}_i+\frac{1}{2}\right).
\end{align}

We will equate the {\em exact}
expectation values of nuclear positions, momenta and
bilinear combinations thereof with those
of a fictitious, non-interacting bosonic system.
Thus if the expectation values $\langle\hat{d}_i^{\dag}\rangle$ and
$\langle\hat{d}_i\rangle$
are known, then expectation values of the displacement
and momentum operators can be reconstructed from
$\langle\hat{\bf u}\rangle=\frac{1}{2}\mathcal{S}^{-1}
 (\langle\hat{d}^{\dag}\rangle^t+\langle\hat{d}\rangle)$ and
$\langle\hat{\bf p}\rangle=\frac{i}{2}\mathcal{T}^{-1}
 (\langle\hat{d}^{\dag}\rangle^t-\langle\hat{d}\rangle)$.
Bilinear expectation values
$\langle\hat{d}_i^{\dag}\hat{d}_j^{\dag}\rangle$,
$\langle\hat{d}_i\hat{d}_j\rangle$ and
$\langle\hat{d}_i^{\dag}\hat{d}_j\rangle$ can be used to
evaluate corresponding products of momentum and position.
For instance
\begin{align}
 \langle\hat{\bf u}\otimes\hat{\bf p}\rangle
 =\frac{i}{4}\mathcal{S}^{-1}\left\langle
 (\hat{d}^{\dag})^t\hat{d}^{\dag}
 -(\hat{d}^{\dag})^t(\hat{d})^t
 -\hat{d}\hat{d}^{\dag}
 +\hat{d}(\hat{d})^t\right\rangle(\mathcal{T}^{-1})^t.
\end{align}
Note that in the unperturbed harmonic oscillator ground state, all these
expectation values are zero.
A further point is that the Hermiticity of the second-quantized
bosonic system described
below renders some of these expectation values inaccessible, one of which
is the nuclear current density. By removing the Hermitian constraint this
restriction is lifted.

\section{Algebraic form of the electron and phonon Kohn-Sham equations}
In this section, the details of the Kohn-Sham Hamiltonian, such as
that in (\ref{KS_fm_r}), are removed and we focus on the algebraic
structure instead. This is done by considering only the matrix elements of
the electron and phonon Hamiltonians. In the following section all matrices
are taken to be finite in size.

\subsection{Kohn-Sham Hamiltonian for electrons}
The most general fermionic Kohn-Sham Hamiltonian of interest here has the form
\begin{align}\label{Hfm_ks}
 \hat{H}_s^{\rm f}=
 \sum_{i,j=1}^{n_{\rm f}}A_{ij}\hat{a}_i^{\dag}\hat{a}_j
 +B_{ij}\hat{a}_i^{\dag}\hat{a}_j^{\dag}
 -B_{ij}^*\hat{a}_i\hat{a}_j,
\end{align}
where $A$ is a Hermitian matrix representing (\ref{KS_fm_r});
$B$ is antisymmetric and
corresponds to the matrix elements of the superconducting pairing potential
$\Delta({\bf r},{\bf r}')$.
The sum runs to the number of fermionic basis vectors $n_{\rm f}$.
The matrix $A$ includes a chemical potential term
$A_{ij}\rightarrow A_{ij}+\mu\delta_{ij}$ which is used to fix the total
electronic number to $N_{\rm e}$.
The Hermitian eigenvalue problem
\begin{align}\label{hm_bog_fm}
 \Bigg(\begin{matrix}
  A &  B \\
  B^{\dag} & -A^*
 \end{matrix}\Bigg)
 \Bigg(\begin{matrix}
  \vec{U}_j \\
  \vec{V}_j
 \end{matrix}\Bigg)
 =\varepsilon_j
 \Bigg(\begin{matrix}
  \vec{U}_j \\
  \vec{V}_j
 \end{matrix}\Bigg)
\end{align}
yields $2n_{\rm f}$ solutions.
However, if $\varepsilon_j$ and $(\vec{U}_j,\vec{V}_j)$ are an eigenpair, then
so are $-\varepsilon_j$ and $(\vec{V}_j^*,\vec{U}_j^*)$.
Now we select $n_{\rm f}$ eigenpairs
with each corresponding to either a positive or negative eigenvalues but
with its conjugate partner not in the set.
This choice will not affect the eventual Kohn-Sham ground state.
Let $U$ and $V$ be the $n_{\rm f}\times n_{\rm f}$ matrices with these
solutions arranged column-wise.
Orthogonality of the vectors is then expressed as
\begin{align}
 \Bigg(\begin{matrix}
  \,U\, &  \,V^*\, \\
  \,V\, &  \,U^*\,
 \end{matrix}\Bigg)^{\dag}
 \Bigg(\begin{matrix}
  \,U\, &  \,V^*\, \\
  \,V\, &  \,U^*\,
 \end{matrix}\Bigg)
 =I,
\end{align}
which implies $U^{\dag}U+V^{\dag}V=I$ and
$U^{\dag}V^*+V^{\dag}U^*=0$. Completeness further implies
$UU^{\dag}+V^*V^t=I$ and $UV^{\dag}+V^*U^t=0$.
The Hamiltonian (\ref{Hfm_ks}) can now
be diagonalized with the aid of $U$ and $V$ via a Bogoliubov transformation:
\begin{align}\label{bog_tfm}
\begin{split}
 \hat{\alpha}_j^{\dag}&=\sum_{i=1}^{n_{\rm f}} U_{ij}\hat{a}_i^{\dag}+V_{ij}\hat{a}_i \\
 \hat{\alpha}_j&=\sum_{i=1}^{n_{\rm f}} U_{ij}^*\hat{a}_i+V_{ij}^*\hat{a}_i^{\dag},
\end{split}
\end{align}
in other words
\begin{align}
 \hat{H}_s=\sum_{i=1}^{n_{\rm f}}\varepsilon_i\hat{\alpha}_i^{\dag}\hat{\alpha}_i+W_0,
\end{align}
where $W_0=-{\rm tr}(V\varepsilon V^{\dag})$. The fermionic algebra
is also preserved for $\hat{\alpha}$:
\begin{align}\label{alpha_acr}
 \bigl\{\hat{\alpha}_i,\hat{\alpha}_j^{\dag}\bigr\}=\delta_{ij} \qquad
 \bigl\{\hat{\alpha}_i,\hat{\alpha}_j\bigr\}=0 \qquad
 \bigl\{\hat{\alpha}_i^{\dag},\hat{\alpha}_j^{\dag}\bigr\}=0.
\end{align}

\subsubsection{Non-interacting ground state}
Given $A$ and $B$, the matrices $U$, $V$ and $\varepsilon$ are fixed by
the Kohn-Sham-Bogoliubov equations (\ref{hm_bog_fm}).
What remains is to construct from these the eigenstates of
(\ref{Hfm_ks}) in the Fock space. To do so, one first needs
to find a normalized vacuum state which is anihilated by all the $\hat{\alpha}_j$.
Here it is (denoted $|\bar{0}\rangle$ so as to distinguish it from the normal vacuum state
$|0\rangle$):
\begin{align}
 |\bar{0}\rangle\equiv\prod_{j=1}^{n_{\rm f}}\hat{U}_j\prod_{k=1}^{n_{\rm f}}
 \hat{a}_k^{\dag}|0\rangle+\prod_{j=1}^{n_{\rm f}}\hat{V}_j^{\dag}|0\rangle,
\end{align}
where $\hat{U}_j\equiv\sum_i U_{ij}^*\hat{a}_i$ and
$\hat{V}_j^{\dag}\equiv\sum_i V_{ij}^*\hat{a}_i^{\dag}$. It is readily
verified that
$\hat{\alpha}_j|\bar{0}\rangle=0$
for all $j$; the vacuum has the correct normalisation
$\langle\bar{0}|\bar{0}\rangle=1$;
and the vacuum energy $\langle\bar{0}|H_s|\bar{0}\rangle=W_0$.
The non-interacting many-body ground state can be constructed
in analogy with the usual fermionic situation. Let $M$ be the number of
$\varepsilon_j<0$, then the ground state
\begin{align}\label{gs_fm}
 |\Phi_0\rangle=\prod_{j=1}^M\hat{\alpha}_j^{\dag}|\bar{0}\rangle,
\end{align}
so that
\begin{align}
 \hat{H}_s|\Phi_0\rangle=E_0^s|\Phi_0\rangle,
\end{align}
where $E_0^s=\sum_{j=1}^M\varepsilon_j+W_0$.

\subsection{Normal and anomalous densities}
To determine the densities, both normal and anomalous, one first has to find
the expectation values of pairs of $\hat{a}$ and $\hat{a}^{\dag}$. These in turn
are linear combinations of expectation values of pairs of $\hat{\alpha}$ and
$\hat{\alpha}^{\dag}$. Using the anti-commutation relations (\ref{alpha_acr})
and remembering that $\hat{\alpha}|\bar{0}\rangle=0$, we get
\begin{align}\label{alpha_mat_1}
 \langle\Phi_0|\hat{\alpha}_i^{\dag}\hat{\alpha}_j|\Phi_0\rangle=
 \begin{cases}
  \delta_{ij} & i,j\le M \\
  0 & i,j>M
 \end{cases} \qquad
 \langle\Phi_0|\hat{\alpha}_i\hat{\alpha}_j^{\dag}|\Phi_0\rangle=
 \begin{cases}
  0 & i,j\le M \\
  \delta_{ij} & i,j>M
 \end{cases}
\end{align}
and
\begin{align}\label{alpha_mat_2}
 \langle\Phi_0|\hat{\alpha}_i^{\dag}\hat{\alpha}_j^{\dag}|\Phi_0\rangle=0
 \qquad
 \langle\Phi_0|\hat{\alpha}_i\hat{\alpha}_j|\Phi_0\rangle=0.
\end{align}
Equations (\ref{bog_tfm}), (\ref{alpha_mat_1}) and (\ref{alpha_mat_2})
give the normal and anomalous density matrices:
\begin{align}\label{dm_fm}
 \langle\Phi_0|\hat{a}_i^{\dag}\hat{a}_j|\Phi_0\rangle=
 \sum_{k=1}^M U_{ik}^*U_{jk}+\sum_{k=M+1}^{n_{\rm f}}V_{ik}V_{jk}^*
\end{align}
and
\begin{align}\label{dma_fm}
 \langle\Phi_0|\hat{a}_i^{\dag}\hat{a}_j^{\dag}|\Phi_0\rangle=
 \sum_{k=1}^M U_{ik}^*V_{jk}+\sum_{k=M+1}^{n_{\rm f}}V_{ik}U_{jk}^*.
\end{align}

\subsubsection{Time evolution}
What remains is to determine how the Kohn-Sham state evolves with time
in the time-dependent density function theory (TDDFT) version of the method.
The form of the ground state equations dictates
that of the time-dependent equations.
Thus if we assume that the matrices
$A$ and $B$ are now functions of time, then the time-dependent generalization
of the orbital equation (\ref{hm_bog_fm}) is
\begin{align}\label{hmt_bog_fm}
 i\frac{\partial}{\partial t}
 \Bigg(\begin{matrix}
  \vec{U}_j \\
  \vec{V}_j
 \end{matrix}\Bigg)
 =
 \Bigg(\begin{matrix}
  A(t) & B(t) \\
  B^{\dag}(t) & -A^*(t)
 \end{matrix}\Bigg)
 \Bigg(\begin{matrix}
  \vec{U}_j \\
  \vec{V}_j
 \end{matrix}\Bigg)
\end{align}
with the Kohn-Sham state given by
$|\Phi(t)\rangle=\prod_{i=1}^M\hat{\alpha}_i^{\dag}(t)|\bar{0}\rangle$. It is
easy to show that this state satisfies
\begin{align}
 i\frac{\partial |\Phi(t)\rangle}{\partial t}
 =\left(\sum_{ij}A_{ij}(t)\hat{a}_i^{\dag}\hat{a}_j
 +B_{ij}(t)\hat{a}_i^{\dag}\hat{a}_j^{\dag}
 -B_{ij}^*(t)\hat{a}_i\hat{a}_j\right)|\Phi(t)\rangle
\end{align}
with $|\Phi(t=0)\rangle=|\Phi_0\rangle$. Note that the number
of `occupied orbitals' $M$ remains constant with time.
Here we have assumed that the system has
evolved from its ground state.

\subsection{Kohn-Sham Hamiltonian for phonons}
The most general bosonic Kohn-Sham Hamiltonian of interest here has the form
\begin{align}\label{Hbs_ks}
 \hat{H}_s^{\rm b}=\sum_{ij}D_{ij}\hat{d}_i^{\dag}\hat{d}_j
 +\tfrac{1}{2}E_{ij}\hat{d}_i^{\dag}\hat{d}_j^{\dag}
 +\tfrac{1}{2}E_{ij}^*\hat{d}_i\hat{d}_j
 +\sum_i F_i\hat{d}_i^{\dag}+F_i^*\hat{d}_i,
\end{align}
where $D$ is Hermitian and contains the kinetic energy operator;
$E$ is a complex symmetric matrix and $F$ is a complex vector.
Note that $\hat{H}_{\rm KS}^{\rm b}$ contains the anomalous terms
$\hat{d}_i^{\dag}\hat{d}_j^{\dag}$
and $\hat{d}_i\hat{d}_j$.
In analogy with the fermionic case, this Hamiltonian can be diagonalized
\begin{align}\label{Hbs_bog}
 \hat{H}_s^{\rm b}=\sum_{i=1}^{n_{\rm b}}
  \omega_i\hat{\gamma}_i^{\dag}\hat{\gamma}_i+\Omega_0
\end{align}
with the Bogoliubov-type transformation
\begin{gather}\label{bog_bs}
\begin{split}
 \hat{\gamma}_j=\sum_{i=1}^{n_{\rm b}} W_{ij}^*\hat{d}_i+X_{ij}^*\hat{d}_i^{\dag}+y_j^* \\
 \hat{\gamma}_j^{\dag}=\sum_{i=1}^{n_{\rm b}} W_{ij}\hat{d}_i^{\dag}+X_{ij}\hat{d}_i+y_j,
\end{split}
\end{gather}
where $W$ and $X$ are complex matrices and $y$ is a complex vector. The index
$j$ runs from $1$ to twice the number of bosonic modes.
Requiring that
$\hat{\gamma}$ and $\hat{\gamma}^{\dag}$ obey bosonic algebra (the
complex numbers $y_j$ obviously commute with themselves and the operators,
maintaining the algebra) yields
\begin{align}
  W^{\dag}W-X^{\dag}X=I \label{wx_cond1} \\
  W^tX-X^tW=0. \label{wx_cond2}
\end{align}
After some manipulation, we
arrive at the Kohn-Sham-Bogoliubov equations for phonons:
\begin{align}\label{hm_bog_bs}
 \Bigg(\begin{matrix}
  D & -E \\
  E^* & -D^*
 \end{matrix}\Bigg)
 \Bigg(\begin{matrix}
  \vec{W}_j \\
  \vec{X}_j
 \end{matrix}\Bigg)
 =\omega_j
 \Bigg(\begin{matrix}
  \vec{W}_j \\
  \vec{X}_j
 \end{matrix}\Bigg).
\end{align}
The above equation can not be reduced to a symmetric eigenvalue problem because the
conditions (\ref{wx_cond1}) and (\ref{wx_cond2})
correspond to the indefinite inner product
$\eta={\rm diag}(1,\ldots,1,-1,\ldots,-1)$. Such matrix Hamiltonians can still
possess real eigenvalues \cite{Sudarshan1961,Mostafazadeh2002}.

\subsubsection{Real case}
We now consider the special case where the matrices $D$ and $E$ are real
symmetric and the vector $F$ is also real.
The bosonic Hamiltonian can be written as
\begin{align}\label{Hbs_ks_r}
 \hat{H}_s^{\rm b}=\sum_{ij}D_{ij}\hat{d}_i^{\dag}\hat{d}_j
 +\tfrac{1}{2}E_{ij}\left(\hat{d}_i^{\dag}\hat{d}_j^{\dag}
 +\hat{d}_i\hat{d}_j\right)
 +\sum_i F_i\left(\hat{d}_i^{\dag}+\hat{d}_i\right).
\end{align}
We now prove that under certain conditions, the matrix equation
(\ref{hm_bog_bs}) always possesses $n_{\rm b}$ solutions which satisfy
(\ref{wx_cond1}) and (\ref{wx_cond2}).
This requires the observation that if the vector $v\equiv(w,x)$ with
eigenvalue $\omega$ is a solution
to (\ref{hm_bog_bs}), then so is $\bar{v}\equiv(x,w)$ with eigenvalue
$-\omega$.

\begin{theorem}
Let
\begin{align*}
 H=
 \Bigg(\begin{matrix}
  D & -E \\
  E & -D
 \end{matrix}\Bigg),
\end{align*}
where $D$ and $E$ are real symmetric $n_{\rm b}\times n_{\rm b}$
matrices. Suppose $H$ has only real, non-degenerate eigenvalues
and every eigenvector $v$ satisfies $v^t\eta v\ne 0$. Then
\renewcommand{\theenumi}{\roman{enumi}}
\begin{enumerate}
\item The eigenvectors of $H$ may be chosen real.
\item The eigenvalue equation (\ref{hm_bog_bs})
has exactly $n_{\rm b}$ solutions
which satisfy the conditions (\ref{wx_cond1}) and (\ref{wx_cond2}).
\end{enumerate}
\end{theorem}
\begin{proof}
The proof that the eigenvectors may be chosen real is straight-forward, so
we now prove the second statement.
Let $v_1$ and $v_2$ be two real eigenvectors of $H$ with corresponding
real eigenvalues $\omega_1$ and $\omega_2$.
Now $Hv_1=\omega_1v_1\Rightarrow \eta Hv_1=\omega_1\eta v_1$
and because $\eta H$ is symmetric we have
$v_1^t\eta H=\omega_1v_1^t\eta$ and thus
$v_1^t\eta Hv_2=\omega_1v_1^t\eta v_2$.
We also have that $Hv_2=\omega_2v_2$ and so
$v_1^t\eta Hv_2=\omega_2v_1^t\eta v_2$.
Subtracting and using the fact that $\omega_1\ne\omega_2$
yields $v_1^t\eta v_2=0$. This is equivalent to the off-diagonal part
of condition (\ref{wx_cond1}). Consider an eigenvector $v=(w,x)$ of $H$.
Now $v^t\eta v\ne 0$, thus if $v^t\eta v<0$ then choose the other eigenvector
$\bar{v}$ for which $\bar{v}^t\eta \bar{v}>0$.
Such an eigenvector can be rescaled arbitrarily
to ensure $v^t\eta v=1$. This corresponds to the diagonal part of
(\ref{wx_cond1}) but is valid for only half of the total number of eigenvectors
since rescaling cannot change the sign of $v^t\eta v$.
These remaining vectors are discarded.
Condition (\ref{wx_cond2}) is trivially satisfied for the diagonal.
For any two vectors $v_i$ and $v_j$ suppose $v_j\ne\bar{v}_i$ then
$\bar{v}_j=v_k$ for some other $k$. The off-diagonal part of condition
(\ref{wx_cond1}) is satisfied for all vectors, thus
$v_i^t\eta v_k=v_i^t\eta \bar{v}_j=0$.
If $v_j=\bar{v}_i$ then one of these vectors will have been discarded.
\end{proof}
The theorem is easily extended to the case where $H$ has degenerate
eigenvalues.
There is no guarantee that the eigenvalues of $H$ are real since the matrix
is not Hermitian. We therefore need additional restrictions on the matrices
$D$ and $E$ to ensure this; the following conditions are sufficient but not
necessary. We use the notation $P\succ 0$ to mean that the symmetric matrix
$P$ is positive definite, and that $P\succ Q$ implies $P-Q\succ 0$.

\begin{theorem}\label{th_LH}
Let $D\succ 0$, and suppose that $E$ is a symmetric matrix. If any of the
following are true then $H$ has real eigenvalues:
\renewcommand{\theenumi}{\roman{enumi}}
\begin{enumerate}
\item $D\succ E D^{-1}E$.\label{pos1}
\item The largest eigenvalue of $(ED^{-1})^2$ is less than $1$.\label{pos2}
\item $z^{\dag}Dz>|z^{\dag}Ez|$ for all $z\in\mathbb{C}^{n_{\rm b}}$.\label{pos3}
\item $E\succ 0$ and $D\succ E$.\label{pos4}
\item $E\succ 0$ and $D^p\succ E^p$, where $p\ge 1$.\label{pos5}
\item $D^2\succ E^2$.\label{pos6}
\end{enumerate}
Furthermore, if all eigenvalues are non-zero then all eigenvectors
satisfy $v^t\eta v\ne 0$.
\end{theorem}
\begin{proof}
Let $\omega$ and $v$ be an eigenvalue and eigenvector of $H$.
The matrix
\begin{align*}
 \eta H=
 \Bigg(\begin{matrix}
  D & -E \\
  -E & D
 \end{matrix}\Bigg)
\end{align*}
is symmetric, therefore both sides of $v^{\dag}\eta H v=\omega v^{\dag}\eta v$
are real. The only requirement for $\omega$ to be real is that
$v^{\dag}\eta H v$ be non-zero, which is ensured so long as $\eta H\succ 0$.
This follows from either of the conditions \ref{pos1} or \ref{pos2}
(see, for example, Ref. \cite{Horn1990}).
Condition \ref{pos3} follows from Theorem 2.1 in Ref \cite{Fitzgerald1977}
and \ref{pos4} follows immediately.
The L\"{o}wner-Heinz theorem \cite{Zhan2002}
reduces condition \ref{pos5} to \ref{pos4}.
Finally, suppose $D^2\succ E^2$ where $E$ may not be positive definite.
$E$ is symmetric therefore $E^2\succ 0$ which means that there exists a symmetric matrix
$e\succ 0$ such that $e^2=E^2$. The L\"{o}wner-Heinz theorem
implies that $D\succ e$, therefore $z^{\dag}Dz > z^{\dag}ez$ for all complex
vectors $z\in\mathbb{C}^{n_{\rm b}}$. $E$ and $e$ can be simultaneously
diagonalized and for each eigenvalue $\lambda$ of $E$
there is a corresponding positive eigenvalue $|\lambda|$ of $e$.
In this eigenvector basis, it is easy to see
that $z^{\dag}ez\ge|z^{\dag}Ez|$ for all $z$ which in turn gives
condition \ref{pos3}, thereby proving \ref{pos6}.
In fact, all of the above conditions imply \cite{Fitzgerald1977} that
$\eta H\succ 0$. Thus if all eigenvalues $\omega\ne 0$ then $v^t\eta v\ne 0$.
\end{proof}

\begin{corollary}\label{cor_psd}
Let $D_0\succ 0$ and $E\succeq 0$ (positive semi-definite) then $D=D_0+E$ yields
real eigenvalues for $H$.
\end{corollary}

\begin{theorem}
Let $D$ be an arbitrary real symmetric matrix and let $f$ be a real function
such that $|f(x)|<|x|$ for all $x\in\mathbb{R}$, then by setting $E=f(D)$
(in the usual `function of matrices' sense \cite{Rinehart1955})
$H$ has real eigenvalues and every eigenvector $v$ satisfies $v^t\eta v\ne 0$.
\end{theorem}
\begin{proof}
We first note that
\begin{align*}
 H^2=
 \Bigg(\begin{matrix}
   D^2-E^2 & [E,D] \\
   [E,D] & D^2-E^2
 \end{matrix}\Bigg).
\end{align*}
It is obvious for any $E=f(D)$ that $[E,D]=0$ and $D^2\succ E^2$.
Therefore all the eigenvalues of $H^2$ are real and positive.
We conclude that the eigenvalues of $H$ are real and non-zero,
thus $v^t\eta v\ne 0$ follows from Theorem \ref{th_LH}.
\end{proof}

\begin{theorem}
Let $D$ be a real symmetric matrix which has no zero eigenvalues and which
commutes with all the
matrices in a group representation $S=\{S_i\}$.
Further suppose that any degenerate eigenvalues of $D$ correspond only to
irreducible representations of $S$ (i.e. there are no accidental degeneracies).
If $E$ is a real symmetric matrix which also commutes with all the matrices
in $S$ then there exists a $\xi>0$ such that if $E\rightarrow \xi E$ then $H(\xi)$
has real eigenvalues.
\end{theorem}
\begin{proof}
From the properties of the determinant applied to blocked matrices, the
eigenvalues of $H^2$ are also the eigenvalues of
$Q\coloneqq D^2-E^2+[E,D]$.
Since $[D,S_i]=[E,S_i]=0$ for all $i$ then $D^2$, $E^2$, $[E,D]$ and thus
$Q(\xi)$ also commute with $S_i$.
Schur's lemma applies equally well to non-Hermitian matrices therefore the
degeneracies of $Q(\xi)$ are not lost as $\xi$ increases.
We also note that the roots of a polynomial depend continuously on its
coefficients and hence the eigenvalues of $Q(\xi)$ depend continuously on $\xi$.
From the conjugate root theorem, if $Q(\xi)$ has a complex eigenvalue then
it must also have its complex conjugate as an eigenvalue.
For sufficiently small $\xi>0$ the eigenvalues of $D^2$ cannot become
complex because this would require lifting of a degeneracy. Also
because of continuity and because
$D^2$ has strictly positive eigenvalues, a sufficiently small
$\xi>0$ will keep them positive. Hence the eigenvalues of $H(\xi)$ are real.
\end{proof}

Once these equations are solved, the vector $y$ is determined from
\begin{align}\label{hm_bog_bs2}
 y=\omega^{-1}\left(W^t-X^t\right)F,
\end{align}
where $\omega={\rm diag}(\omega_1,\ldots,\omega_{n_{\rm b}})$.
The constant term in (\ref{Hbs_bog}) given by
\begin{align}
 \Omega_0=-{\rm tr}\left(X\omega X^{\dag}\right)-y^{\dag}\omega y.
\end{align}
%An important observation must be made about (\ref{Hbs_bog}): any non-positive
%eigenvalues $\omega_i$ will result in infinite negative energy density.
%Hence, if any eigenvalues are less than or equal to zero, then a
%`chemical potential' term should be added to $D$ to shift all eigenvalues
%upwards.

\subsubsection{Existence and nature of the vacuum state}
We now show that the state which is annihilated by all the $\hat{\gamma}_i$
exists. Let
\begin{align}
 \hat{w}_j\coloneqq\sum_{i=1}^{n_{\rm b}} W_{ij}^*\hat{d}_i \qquad
 \hat{x}_j^{\dag}\coloneqq\sum_{i=1}^{n_{\rm b}} X_{ij}^*\hat{d}_i^{\dag}
\end{align}
then
\begin{align}
 \left[\hat{w}_j,\hat{x}_j^{\dag}\right]=\sum_{i=1}^{n_{\rm b}} W_{ij}^*X_{ij}^*
 \eqqcolon \tau_j.
\end{align}
Now consider the eigenvalue equation
\begin{align}\label{eig_bs}
 \left(\hat{w}_j+\hat{x}_j^{\dag}\right)|\bar{0}_j\rangle=-y_j^*|\bar{0}_j\rangle.
\end{align}
Using the ansatz
\begin{align}\label{coh_bs}
 |\bar{0}_j\rangle=\sum_{n=0}^{\infty}
 \frac{\kappa_n^j}{n!}(\hat{x}_j^{\dag})^n|0\rangle,
\end{align}
we obtain a recurrence relation
\begin{align}
 \kappa_n^j=\left[-y_j^*\kappa_{n-1}^j-(n-1)\kappa_{n-2}^j\right]/\tau_j
\end{align}
with $y_j^*\kappa_0^j=-\kappa_1^j\tau_j$ and $\kappa_0^j$ chosen so that
$\langle\bar{0}_j|\bar{0}_j\rangle=1$. Note that if
$\kappa_n^j=1$ for all $n$ then (\ref{coh_bs}) is a coherent state.
The vacuum state
\begin{align}\label{gs_bs}
 |\bar{0}\rangle=\zeta\hat{S}\bigotimes_{j=1}^{n_{\rm b}}|\bar{0}_j\rangle,
\end{align}
where $\zeta$ is a normalization constant and $\hat{S}$ is the symmetrizing operator,
is annihilated by all $\hat{\gamma}_j$ and, because $\omega_j>0$ for all $j$,
is also the bosonic Kohn-Sham
ground state,
which is the lowest energy Fock space eigenstate of (\ref{Hbs_ks}), as required.

\subsubsection{Phononic observables and time evolution}
To make the theory useful, observables which are products of the original
$c_i$ and $c_i^{\dag}$ operators have to be computed. After some straight-forward
algebra one finds that linear operators may be evaluated using
\begin{align}\label{bs_obs1}
 Y_i\coloneqq\langle\bar{0}|\hat{d}_i|\bar{0}\rangle=
 \langle\bar{0}|\hat{d}_i^{\dag}|\bar{0}\rangle^*=
 \sum_{j=1}^{n_{\rm b}} X_{ij}^*y_j-W_{ij}y_j^*.
\end{align}
Observables which are quadratic are more complicated:
\begin{align}\label{bs_obs2}
\begin{split}
 \langle\bar{0}|\hat{d}_i^{\dag}\hat{d}_j|\bar{0}\rangle=
 Y_i^*Y_j+\left(XX^{\dag}\right)_{ij} \qquad
 \langle\bar{0}|\hat{d}_i\hat{d}_j^{\dag}|\bar{0}\rangle=
 Y_iY_j^*+\left(WW^{\dag}\right)_{ij} \\
 \langle\bar{0}|\hat{d}_i^{\dag}\hat{d}_j^{\dag}|\bar{0}\rangle=
 Y_i^*Y_j^*-\left(XW^{\dag}\right)_{ij} \qquad
 \langle\bar{0}|\hat{d}_i\hat{d}_j|\bar{0}\rangle=
 Y_iY_j-\left(WX^{\dag}\right)_{ij}.
\end{split}
\end{align}
The extension to the time-dependent case follows the same procedure as that for
fermions, namely that the matrices and vector $D$, $E$ and $F$
in (\ref{Hbs_ks}) become time-dependent as, consequently,
do $\hat{\gamma}_i^{\dag}$ and $|\bar{0}\rangle$ after solving the equation of
motion
\begin{align}\label{hmt_bog_bs1}
 i\frac{\partial}{\partial t}
 \Bigg(\begin{matrix}
  \vec{W}_j \\
  \vec{X}_j
 \end{matrix}\Bigg)=
 \Bigg(\begin{matrix}
  D(t) & -E(t) \\
  E^*(t) & -D^*(t)
 \end{matrix}\Bigg)
 \Bigg(\begin{matrix}
  \vec{W}_j \\
  \vec{X}_j
 \end{matrix}\Bigg).
\end{align}
This time evolution is not unitary but rather pseudo-unitary
\cite{Mostafazadeh2002b} and will not preserve ordinary
vector lengths in general but will preserve the indefinite inner product.
The vector $y$ can be determined analogously from
\begin{align}\label{hmt_bog_bs2}
 i\frac{\partial y}{\partial t}=\left(W^t(t)-X^t(t)\right)F(t).
\end{align}
Evolving (\ref{hmt_bog_bs1}) and (\ref{hmt_bog_bs2}) in time
is equivalent to doing the same for the second-quantized Hamiltonian and
the Fock space state vector:
\begin{align}
 i\frac{\partial |\Psi(t)\rangle}{\partial t}
 =\left(\sum_{ij}D_{ij}(t)\hat{d}_i^{\dag}\hat{d}_j
 +\tfrac{1}{2}E_{ij}(t)\hat{d}_i^{\dag}\hat{d}_j^{\dag}
 +\tfrac{1}{2}E_{ij}^*(t)\hat{d}_i\hat{d}_j
 +\sum_i F_i(t)\hat{d}_i^{\dag}+F_i^*(t)\hat{d}_i\right)|\Psi(t)\rangle.
\end{align}

\subsubsection{Numerical aspects}
In order to determine the phonon ground state or perform
time-evolution with (\ref{hmt_bog_bs1}) for real systems,
we require a numerical algorithm for finding the eigenvalues and eigenvectors
of (\ref{hm_bog_bs}). This is not a symmetric or Hermitian problem and
while a general non-symmetric eigenvalue solver could be employed,
a simple modification of Jacobi's method can be used to diagonalize the
matrix efficiently.

Let $G(i,j,\theta)$ be a Givens rotation matrix, i.e. for $i<j$,
$G_{kk}=1$ for $k\ne i,j$, $G_{kk}=\cos\theta$ for $k=i,j$,
$G_{ji}=-G_{ij}=\sin\theta$ and zero otherwise. Further define the
hyperbolic Givens rotation, $G^{\rm h}(i,j,\theta)$, which is the same except that
$G_{kk}=\cosh\theta$ and $G_{ji}=G_{ij}=\sinh\theta$.
The Givens and hyperbolic Givens rotations can be combined to diagonalize the
matrix in (\ref{hm_bog_bs}).
For $i<j$ where $1<j\le 2N_{\rm b}$ we can define a combined Givens
rotation, $G^{\rm c}(i,j,\theta)$, as
$G^{\rm c}=G^{\rm h}$ for $i\le N_{\rm b}$ and $j>N_{\rm b}$;
and
$G^{\rm c}(i,j,\theta)=G(i,j,\theta)G(i+N_{\rm b},j+N_{\rm b},\theta)$
for $i,j\le N_{\rm b}$.
\begin{definition}
A pair of real, symmetric matrices $A$, $B$ is called positive definite
if there exists a real $\mu$ such that $A-\mu B$ is positive definite.
\end{definition}
\begin{theorem}
Let $\eta H$ and $\eta$ be a positive definite pair. Then applying the
combined Givens rotations to $H$
with row-cyclic strategy results in convergence to
a diagonal matrix.
\end{theorem}
\noindent See Veseli\'{c}\cite{veselic93} for proof.

\subsubsection{Solids}
Solid state calculations normally use periodic boundary conditions and
Bloch orbitals. Phonon displacements are of the form
\begin{align}\label{uphn}
\mathcal{U}_{n{\bf q}}({\bf R})=N_q^{-1/2}
 2^{-\frac{1}{2}}\nu^{\frac{1}{2}}{\bf e}_{n{\bf q}}e^{i{\bf q}\cdot{\bf R}},
\end{align}
where ${\bf q}$ is a reciprocal lattice vector, $\alpha$ labels a phonon
branch, ${\bf R}$ is a primitive lattice vector
and ${\bf e}_{n{\bf q}}$ is determined along with
$\nu_{n{\bf q}}$
by solving (\ref{evphn}) for each ${\bf q}$-vector individually.
These displacements are thus complex-valued but by
noting that $\nu_{n-{\bf q}}=\nu_{n{\bf q}}$ and
${\bf e}_{n-{\bf q}}={\bf e}_{n{\bf q}}^*$
we can form their real-valued counterparts
\begin{align*}
 \mathcal{U}_{n{\bf q}}^{(+)}({\bf R})=
 \frac{1}{\sqrt{2}}\left(\mathcal{U}_{n{\bf q}}({\bf R})
 +\mathcal{U}_{n-{\bf q}}({\bf R})\right) \qquad
 \mathcal{U}_{n{\bf q}}^{(-)}({\bf R})=
 \frac{-i}{\sqrt{2}}\left(\mathcal{U}_{n{\bf q}}({\bf R})
 -\mathcal{U}_{n-{\bf q}}({\bf R})\right).
\end{align*}
These are the displacements to which $\hat{d}_i$ and $\hat{d}_i^{\dag}$
refer and will thus keep the phonon Hamiltonian in (\ref{Hbs_ks_r}) real.
An approximate
electron-phonon vertex is obtained as a by-product of a phonon calculation:
\begin{align}
 \Gamma_{i{\bf k}+{\bf q},j{\bf k},n{\bf q}}
 =\frac{1}{2}\langle\varphi_{j{\bf k}+{\bf q}}|
 \partial\hat{V_s}/\partial\mathcal{U}_{n{\bf q}}|\varphi_{i{\bf k}}\rangle
\end{align}
where $\hat{V_s}$ is the Kohn-Sham potential and the derivative is with
respect to the magnitude of the displacement in (\ref{uphn}).
This is not Hermitian in the indices $i$ and $j$ because the potential
derivative corresponds to a complex displacement.
The vertex associated with $\mathcal{U}_{n{\bf q}}^{(\pm)}$ has the form
\begin{align}\label{hvertex}
 \bordermatrix{
    & {\bf k}-{\bf q} & {\bf k} & {\bf k}+{\bf q} \cr
  {\bf k}-{\bf q} & 0 & \Gamma_{n-{\bf q}} & 0 \cr
  {\bf k} & \Gamma_{n-{\bf q}}^{\dag} & 0 & \Gamma_{n{\bf q}} \cr
  {\bf k}+{\bf q} & 0 & \Gamma_{n{\bf q}}^{\dag} & 0}
\end{align}
which is a Hermitian matrix for all ${\bf q}$ and $n$.

One final point regarding solids is the requirement of keeping the
electronic densities
lattice periodic. This implies that the potentials $A$ and $B$ should only
couple the Bloch vector ${\bf k}$ with itself.

\section{Mean-field functionals}
The final (and possibly most difficult) step in this
theory is the determination of
potentials represented by the matrices $A$, $B$, $D$, $E$ and vector $F$.
In principle, these are chosen to reproduce the exact conditional density
$\rho_{\dub{\bf R}^0}({\bf r},t)$ in (\ref{rho_0}) as well as the
phononic expectation values $\langle\hat{d}_i\rangle$,
$\langle\hat{d}_i^{\dag}\hat{d}_j\rangle$, etc., which themselves reproduce
exact nuclear positions, momenta and so on.
In practice, these potentials need to be approximated and here we will employ
a simple mean-field approach by considering the lowest order diagrams
which enter the self-energy. These are plotted in Fig. \ref{fig_fd} and
involve the normal and anomalous, Kohn-Sham
electronic Green's functions
$i\mathcal{G}_{ij}(t,t')
 =\langle\Phi_0|T[\hat{a}_i(t)\hat{a}_j^{\dag}(t')]|\Phi_0\rangle$ and
$i\mathcal{F}_{ij}(t,t')
 =\langle\Phi_0|T[\hat{a}_i^{\dag}(t)\hat{a}_j^{\dag}(t')]|\Phi_0\rangle$, etc.,
as well as the phonon propagators
$i\mathcal{C}_i(t)=\langle\bar{0}|\hat{d}_i^{\dag}(t)|\bar{0}\rangle$, etc.
and
$i\mathcal{D}_{ij}(t,t')
 =\langle\bar{0}|T[\hat{d}_i(t)\hat{d}_j^{\dag}(t')]|\bar{0}\rangle$, etc.
These
quantities are evaluated around their respective Kohn-Sham ground states,
(\ref{gs_fm}) and (\ref{gs_bs}).
The quantity $A_0$ is given by the matrix elements of the single particle
Hamiltonian in (\ref{KS_fm_r}) without $V_{\rm fmc}$, and $D_0=\nu$.
\vskip 0.5cm
\begin{figure}[ht]
\centerline{\includegraphics[width=\textwidth]{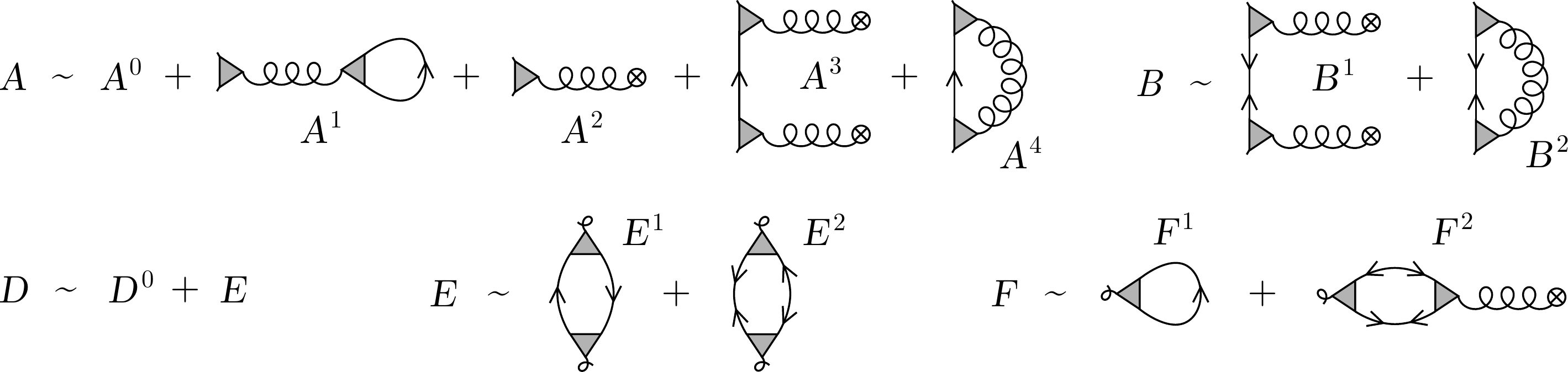}}
\caption[]{The lowest order contributions to the self-energy
 from the vertex
 $\Gamma=\raisebox{-2pt}{\mbox{\includegraphics[height=11pt]{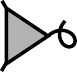}}}$,
 the normal and anomalous Green's functions
 $\mathcal{G}={\mbox{\includegraphics[height=6pt]{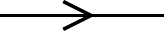}}}$ and
 $\mathcal{F}={\mbox{\includegraphics[height=6pt]{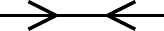}}}$,
 and the phonon propagators
 $\mathcal{C}={\mbox{\includegraphics[height=7pt]{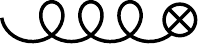}}}$ and
 $\mathcal{D}={\mbox{\includegraphics[height=7pt]{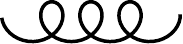}}}$.
 These are evaluated in the static limit as mean-field potentials for
 $A$, $B$, $D$, $E$ and $F$.}\label{fig_fd}
\end{figure}
Explicit expressions for the potentials are found by substituting
instantaneous densities or density matrices
of the electrons and phonons for the retarded correlation
functions in the diagrams.
For example, the electronic state would be affected by the phonon system via
the expectation values of the phonon operators, yielding a contribution to $A$:
\begin{align}
 A_{ij}^2(t)=\sum_k\Gamma_{ijk}\left(\langle\hat{d}_k^{\dag}\rangle_t
 +\langle\hat{d}_k\rangle_t\right),
\end{align}
where the expectation values are evaluated with (\ref{bs_obs1}) and
$\Gamma_{ijk}$ is shorthand for the vertex in (\ref{hvertex}).
At first glance, the matrix $A^3$ appears to be an improper part of the
self-energy which is already accounted for by $A^2$. Such a term is still valid
for solids with since $A^2$ can only ever couple ${\bf k}$ with itself.
However, the Green's function line in $A^3$ can carry momentum ${\bf q}\ne 0$
and yet have the potential preserve lattice periodicity.

The mean-field potential that gives rise to superconductivity is a little more
complicated:
\begin{align}
 B_{ij}^2(t)=-\sum_{klmn}\Gamma_{ikl}\Gamma_{mjn}
 \left(\langle\hat{a}_m^{\dag}\hat{a}_k^{\dag}\rangle_t
 +\langle\hat{a}_m\hat{a}_k\rangle_t\right)
 \left(\langle\hat{d}_l^{\dag}\hat{d}_n^{\dag}\rangle_t
 +\langle\hat{d}_l^{\dag}\hat{d}_n\rangle_t
 +\langle\hat{d}_l\hat{d}_n^{\dag}\rangle_t
 +\langle\hat{d}_l\hat{d}_n\rangle_t\right),
\end{align}
where the density matrices are determined from (\ref{dma_fm}) and
(\ref{bs_obs2}).
The potential represented by $F$ would be
\begin{align}
 F_k^1(t)=\sum_{ij}\Gamma_{ijk}\gamma_{ij}(t)
\end{align}
where $\gamma_{ij}(t)=\langle\hat{a}_i^{\dag}\hat{a}_j\rangle_t$
is the electronic one-reduced density matrix calculated using (\ref{dm_fm}).
The matrix $E^1$ is evaluated as:
\begin{align}\label{mat_e1}
 E_{ij}^1(t)=\sum_{klmn}\Gamma_{kli}\Gamma_{mnj}
 \gamma_{kn}(t)\gamma_{ml}(t).
\end{align}
This matrix should be positive semi-definite in order to satisfy
Corollary \ref{cor_psd} and guarantee real eigenvalues for the bosonic
Hamiltonian in (\ref{Hbs_ks_r}).
\begin{lemma}
 The matrix $E^1$ is positive semi-definite.
\end{lemma}
\begin{proof}
 We first note that $\Gamma_{kli}=\Gamma_{lki}^*$ for all $i$, i.e. $\Gamma$ is
 Hermitian in the electronic indices.
 Since $\Gamma_{kli}\Gamma_{mnj}\gamma_{kn}\gamma_{ml}$ and
 $\Gamma_{lki}\Gamma_{nmj}\gamma_{lm}\gamma_{nk}
 =\Gamma_{kli}^*\Gamma_{mnj}^*\gamma_{ml}^*\gamma_{kn}^*$
 both appear in the sum in (\ref{mat_e1}) then
 $E^1$ must be real and symmetric.
 Let $v$ be a real vector of the same dimension as $E^1$,
 then $R_{kl}\equiv\sum_i v_i \Gamma_{kli}$ is also Hermitian.
 The quantity $s\equiv v^t E^1 v$ can be written as
 $s={\rm tr}(R^{\dag}\gamma R^{\dag}\gamma)$.
 Let $U$ be the unitary transformation that diagonalizes $\gamma$ and
 define ${\rm diag}(\tilde{\gamma})\equiv U^{\dag}\gamma U$ and
 $\tilde{R}\equiv U^{\dag}RU$, then
 $s={\rm tr}(\tilde{R}^{\dag}\tilde{\gamma}\tilde{R}^{\dag}\tilde{\gamma})$
 is left invariant.
 One of the $N$-representable properties\cite{coleman63} of $\gamma$
 is that its eigenvalues satisfy $0\le\tilde{\gamma}_i\le 1$.
 Then $s=\sum_{kl}|\tilde{R}_{kl}|^2\tilde{\gamma}_k\tilde{\gamma}_l\ge 0$.
 Since $v$ was chosen arbitrarily we conclude that
 $E^1$ is positive semi-definite.
\end{proof}

\section{Summary}
We have defined Kohn-Sham equations for fermions and bosons which are
designed to reproduce conditional electronic densities as well as
expectation values of the phonon creation and annihilation operators.
Sufficient conditions which guarantee real eigenvalues for the bosonic system
were found.
In practice, the potential matrix elements $A$, $B$, $D$, $E$ and $F$
can be approximated using mean-field potentials
inspired from a diagrammatic expansion of the self-energy. The electron and
phonon density matrices are determined either self-consistently in a
ground state calculation or via simultaneous propagation in the time-dependent
case. Any solution obtained in this way is thus non-perturbative.
These equations can be implemented in both finite and solid-state codes
using quantities determined from linear-response phonon calculations.

\section*{Acknowledgments}
We would like to thank James Annett for pointing out the similarity of
our bosonic analysis to that in Ref. \cite{Colpa1978}. We acknowledge
DFG for funding through SPP-QUTIF and SFB-TRR227.

\bibliographystyle{unsrt}
%\bibliography{bibliography}

\end{document}